\documentclass[a4paper]{article}
\usepackage{bm}
\usepackage{amsfonts}
\usepackage[english]{inputenc}
\usepackage{amsmath}
\usepackage{amssymb}
\usepackage{amsthm}
\usepackage{authblk}
\usepackage{color}

\begin{document}

\newtheorem{definition}{Definiton}[section]
\newtheorem{theorem}{Theorem}[section]

\title{Classification  of the static and  asymptotically flat  Einstein-Maxwell-dilaton spacetimes with a photon sphere}
\author[1,2]{Stoytcho Yazadjiev\thanks{yazad@phys.uni-sofia.bg} }
\author[1]{Boian Lazov\thanks{boian\_lazov@phys.uni-sofia.bg}}
\affil[1]{Department of Theoretical Physics, Faculty of Physics, Sofia University, Sofia 1164, Bulgaria}
\affil[2]{Theoretical Astrophysics, Eberhard-Karls University of T\"ubingen, T\"ubingen 72076, Germany}

\maketitle

\begin{abstract}
We consider the problem for the  classification  of static and asymptotically flat Einstein-Maxwell-dilaton  spacetimes with a photon sphere.
It is first proven that the photon spheres in  Einstein-Maxwell-dilaton gravity have constant mean and constant scalar curvature. Then we derive some
relations between the mean curvature and the physical characteristics of the photon spheres. Using further the symmetries of the dimensionally reduced
Einstein-Maxwell-dilaton field equations we show that the lapse function, the electrostatic potential and the dilaton field are functionally dependent
in the presence of a photon sphere.  Using all this we prove  the main classification theorem by explicitly constructing all Einstein-Maxwell-dilaton  solutions possessing a non-extremal
photon sphere.
\end{abstract}

\section{Introduction}

\indent Photon spheres are a very important result in the gravitational theories \cite{Iyer1985}-\cite{Baldiotti2014}. They are closely connected to the strong gravitational
lensing \cite{Virbhadra1998}-\cite{Bozza2010} and the black hole quasi-normal modes \cite{Cardoso2009}-\cite{Decanini2010}, and provide a bridge between them \cite{Stefanov2010}-\cite{Wei2014}. Because of this photon spheres are of great importance to astrophysics. Furthermore, photon spheres have very interesting mathematical properties in spite of their simple definition.\\
\indent There is a physical and mathematical resemblance between event horizons and photon spheres, which focuses attention on the question whether spacetimes can be classified using the existence of a photon sphere instead of an event horizon. This has already been done for the vacuum Einstein equations \cite{Cederbaum2014},\cite{Cederbaum2015a}, Einstein-scalar field equations \cite{Yazadjiev2015} and Einstein-Maxwell equations \cite{Yazadjiev2015a,Cederbaum2015b}, although in general photon sphere uniqueness is a harder problem than event horizon uniqueness, because a larger class of solutions possesses the former compared to the latter.\\
\indent The present paper aims to prove a classification  theorem for the static and asymptotically flat solutions of the Einstein-Maxwell-dilaton (EMD) field equations, containing a photon sphere. This comes as a logical continuation of our previous works \cite{Yazadjiev2015, Yazadjiev2015a}, which proved similar theorems for the  cases of Einstein-scalar field and Einstein-Maxwell equations. Some sectors of  EMD gravity  naturally arise in the context of  low energy string theory \cite{Gibbons1988, Garfinkle1992}, Kaluza-Klein gravity \cite{Maison1979} and also in some theories with gradient spacetime torsion \cite{Hojman1978}.

\section{Definitions and preparation}

\indent Let us denote our spacetime manifold by $(\mathfrak{L}^4,\mathfrak{g})$. The field equations for Einstein-Maxwell-dilaton  gravity are as follows:
\begin{align}
&\mathfrak{R}_{\mu\nu}=2\,{}^{\mathfrak{g}}\!\nabla_{\mu}\varphi\,{}^{\mathfrak{g}}\!\nabla_{\nu}\varphi+ 2e^{-2\alpha\varphi}\left( F_{\mu\beta}F^{\ \beta}_{\nu} -\frac{\mathfrak{g}_{\mu\nu}}{4}F_{\beta\gamma}F^{\beta\gamma} \right),\label{ems1}\\
&{}^{\mathfrak{g}}\!\nabla_{[\beta}F_{\mu\nu]}=0,\label{ems2}\\
&{}^{\mathfrak{g}}\!\nabla_{\beta}\left( e^{-2\alpha\varphi}F^{\beta\mu} \right)=0,\label{ems3}\\
&{}^{\mathfrak{g}}\!\nabla_{\beta}\,{}^{\mathfrak{g}}\!\nabla^{\beta}\varphi=-\frac{\alpha}{2}e^{-2\alpha\varphi}F_{\mu\nu}F^{\mu\nu}.\label{ems4}
\end{align}

\noindent We are using the usual notation. Namely, $\mathfrak{R}_{\mu\nu}$ is the spacetime Ricci tensor, $\varphi$ is the dilaton field, $F_{\mu\nu}$ is the Maxwell field and $\alpha$ denotes the dilaton coupling constant.\\
\indent Static spacetimes can be decomposed as follows. There exist a smooth Riemannian manifold $(M^3,g)$ and a smooth lapse function $N:M^3\longrightarrow \mathbb{R}^+$ such that
\begin{align}
\mathfrak{L}^4=\mathbb{R}\times M^3,\ \mathfrak{g}=-N^2\mathrm{d}t^2+g.\label{statg}
\end{align}

\indent Staticity of the Maxwell and scalar fields is defined using the timelike Killing vector $\xi=\frac{\partial}{\partial t}$,
\begin{align}
&\mathcal{L}_{\xi}F=0,\\
&\mathcal{L}_{\xi}\varphi=0.
\end{align}

\noindent We will be considering the case with $\iota_{\xi}\star F=0$, i. e. the purely electric case.\\
\indent We will be proving our theorem for asymptotically flat spacetimes, which are defined in the following way: A spacetime is asymptotically flat if there exists a compact set $K\in M^3$ such that $M^3\setminus K$ is diffeomorphic to $\mathbb{R}^3\setminus \bar{B}$, where $\bar{B}$ is the closed unit ball centered at the origin of $\mathbb{R}^3$ and such that
\begin{align}
g=\delta+O(r^{-1}),\ N=1-\frac{M}{r}+O(r^{-2}),
\end{align}

\noindent with respect to the radial coordinate on $\mathbb{R}^3$. In the above expression $M$ denotes the mass. Asymptotic expansions of the Maxwell and scalar fields will be given later.\\
\indent We will now define the photon surface.
\begin{definition}
An embedded timelike hypersurface $(P^3,p)\hookrightarrow(\mathfrak{L}^4,\mathfrak{g})$ is called a photon surface if any null geodesic initially tangent to $P^3$ remains tangent to $P^3$ as long as it exists.
\end{definition}

\noindent The definition of the photon sphere is naturally modified keeping in mind the cases with Maxwell and scalar fields.
\begin{definition}
Let $(P^3,p)\hookrightarrow(\mathfrak{L}^4,\mathfrak{g})$ be a photon surface. Then $P^3$ is called a photon sphere if the lapse function $N$ is constant on $P^3$ and the one-forms $\iota_{\xi}F$ and $\mathrm{d}\varphi$ are normal to $P^3$.
\end{definition}

\indent Additionally, we assume that the lapse function $N$ regularly foliates the exterior to the photon sphere region of spacetime, which has spatial part denoted by $M^3_{\mathrm{ext}}$, i. e.
\begin{align}
\rho^{-2}=g({}^g\!\nabla N,{}^g\!\nabla N)\ne 0
\end{align}

\noindent outside the photon sphere. $M^3$ has as inner boundary the intersection $\Sigma$ of the outermost photon sphere with the time slice $M^3$. $\Sigma$ is given by $N=N_0$ for some $N_0\in \mathbb{R}^+$. Our assumption implies that all level sets $N=\mathrm{const}$, including $\Sigma$, are topological spheres and $M^3_{\mathrm{ext}}$ is topologically $S^2\times \mathbb{R}$.\\
\indent The electric field one-form $E$ is defined by
\begin{align}
E=-\iota_{\xi}F
\end{align}

\noindent and it satisfies $\mathrm{d}E=0$, which follows from the field equations and the staticity of the Maxwell field. As a consequence of the fact that $M^3_{\mathrm{ext}}$ is simply connected there exists an electric potential $\Phi$ such that $E=\mathrm{d}\Phi$. With the electric field one-form we can write explicitly an expression for $F$,
\begin{align}
F=-N^{-2}\xi\wedge\mathrm{d}\Phi.
\end{align}

\noindent From the definition of the photon sphere, $E$ is normal to it and thus $\Phi$ is constant on $P^3$.\\
\indent We will now write the asymptotic expansions of the scalar field $\varphi$ and the electric potential $\Phi$,
\begin{align}
&\varphi=\varphi_{\infty}-\frac{q}{r}+O(r^{-2}),\\
&\Phi=\Phi_{\infty}+\frac{Q}{r}+O(r^{-2}),
\end{align}

\noindent where $q$ is the scalar charge and $Q$ is the electric charge. We set $\varphi_{\infty}=0$ and $\Phi_{\infty}=0$.

\indent Using the form of the metric (\ref{statg}) we can obtain the dimensionally reduced EMD equations,
\begin{align}
&{}^g\!\Delta N=N^{-1}e^{-2\alpha\varphi}\,{}^g\!\nabla_i\Phi\,{}^g\!\nabla^i\Phi,\label{emsr1}\\
&{}^g\!R_{ij}=2\,{}^g\!\nabla_i\varphi\,{}^g\!\nabla_j\varphi+N^{-1}\,{}^g\!\nabla_i\,{}^g\!\nabla_jN\label{emsr2}\\
&\ \ \ \ \ \ \ \ +N^{-2}e^{-2\alpha\varphi}( g_{ij}\,{}^g\!\nabla_k\Phi\,{}^g\!\nabla^k \Phi -2\,{}^g\!\nabla_i\Phi \,{}^g\!\nabla_j\Phi ),\notag\\
&{}^g\!\nabla_i(N^{-1}e^{-2\alpha\varphi}\,{}^g\!\nabla^i\Phi)=0,\label{emsr3}\\
&{}^g\!\nabla_i(N\,{}^g\!\nabla^i\varphi)=\alpha N^{-1}e^{-2\alpha\varphi}\,{}^g\!\nabla_i\Phi\,{}^g\!\nabla^i\Phi.\label{emsr4}
\end{align}

\indent Using the maximum principle for elliptic partial differential equations, from the asymptotic behaviour of $N$ for $r\to\infty$ follows that the values of $N$ on $M^3_{\mathrm{ext}}$ satisfy
\begin{align}
N_0\le N<1.
\end{align}

\section{Some important results for the extrinsic and intrinsic geometry of the photon spheres}

\indent We will present here some results that will be used in the proof of the main theorem later. We start by using a theorem by Claudel-Virbhadra-Ellis \cite{Claudel2001}  and Perlick \cite{Perlick2005}.
\begin{theorem}
Let $(P^3,p)\hookrightarrow(\mathfrak{L}^4,\mathfrak{g})$ be an embedded timelike hypersurface. Then $P^3$ is a photon surface if and only if it is totally umbilic (i. e. its second fundamental form is pure trace).
\end{theorem}

\noindent This theorem allows us to write the second fundamental form of $P^3$ as $\mathfrak{h}=\frac{\mathfrak{H}}{3}p$, where $\mathfrak{H}$ is the mean curvature of $P^3$. This will be used to prove the following theorem:
\begin{theorem}
Let $(\mathfrak{L}^4,\mathfrak{g},F,\varphi)$ be a static, asymptotically flat spacetime, satisfying the EMD equations (\ref{emsr1}-\ref{emsr4}) and possessing a photon sphere $(P^3,p)\hookrightarrow(\mathfrak{L}^4,\mathfrak{g})$. Then $P^3$ has constant mean curvature (CMC) and constant scalar curvature (CSC).
\end{theorem}

\begin{proof}
To prove the theorem first we have to use the Codazzi equation for $(P^3,p)\hookrightarrow(\mathfrak{L}^4,\mathfrak{g})$ with unit normal $\nu$. Straightforward calculations (see for example \cite{Cederbaum2014}, \cite{Yazadjiev2015}, \cite{Yazadjiev2015a}) lead to
\begin{align}
\mathfrak{R}(Y,\nu)=(1-3)Y\left( \frac{\mathfrak{H}}{3} \right)=0.
\end{align}

\noindent We have used the field equations (\ref{ems1}) to calculate $\mathfrak{R}(Y,\nu)$ with $Y$ being a tangent vector to $P^3$. This shows that $P^3$ has CMC.\\
\indent Next we will use the contracted Gauss equation, again for $(P^3,p)\hookrightarrow(\mathfrak{L}^4,\mathfrak{g})$. After some simple calculations (similar to \cite{Yazadjiev2015a}) we arrive at
\begin{align}
{}^p\!R=\frac{2}{3}\mathfrak{H}^2-2({}^{\mathfrak{g}}\!\nabla_{\nu}\varphi)^2+e^{-2\alpha\varphi}\frac{2}{N^2}E_{\nu}^2,\label{CSCP}
\end{align}

\noindent where $E_{\nu}=\iota_{\nu}E$. Now to prove that $P^3$ has CSC we need to prove that $E_{\nu}=\mathrm{const}$ and ${}^{\mathfrak{g}}\!\nabla_{\nu}\varphi=\mathrm{const}$ on $P^3$. Below we show that $\nu(N)$ is constant on $P^3$ and  in the next section we show that $\varphi$ and $\Phi$ are functions of $N$. Therefore $E_{\nu}$ and ${}^{\mathfrak{g}}\!\nabla_{\nu}\varphi$ are constant on $P^3$. With this the proof is complete.
\end{proof}

\indent Next we will obtain relations for the mass, the electric charge and the scalar charge on the photon sphere. Integrating (\ref{emsr3}) on $M^3_{\mathrm{ext}}$ we get
\begin{align}
Q=-\frac{1}{4\pi}\int_{\Sigma}N^{-1}e^{-2\alpha\varphi}\,{}^g\!\nabla^i \Phi\mathrm{d}\Sigma_i.
\end{align}

\noindent The same integration is done for (\ref{emsr1}) and (\ref{emsr4}). Using the above expression for the electric charge $Q$, this gives the following expressions for the mass $M$ and the scalar charge $q$:
\begin{align}
&M=M_0+\Phi_0Q,\label{mass}\\
&q=q_0+\alpha \Phi_0 Q\label{sccharge},
\end{align}

\noindent where $M_0$ is the mass of the photon sphere, defined by
\begin{align}
M_0=\frac{1}{4\pi}\int_{\Sigma}{}^g\!\nabla^iN\mathrm{d}\Sigma_i,
\end{align}

\noindent and $q_0$ is the scalar charge of the photon sphere, defined by
\begin{align}
q_0=\frac{1}{4\pi}\int_{\Sigma}N\,{}^g\!\nabla^i\varphi\mathrm{d}\Sigma_i.
\end{align}

\indent Our next step is to compute the second fundamental form $h$ of $(\Sigma,\sigma)\hookrightarrow(M^3,g)$ with unit normal $\nu$. Let $X,Y\in \Gamma(T\Sigma)$. Then
\begin{align}
h(X,Y)=g({}^g\!\nabla_X\nu,Y)=\mathfrak{g}({}^{\mathfrak{g}}\!\nabla_X\nu,Y)=\mathfrak{h}(X,Y)=\frac{\mathfrak{H}}{3}p(X,Y)=\frac{\mathfrak{H}}{3}\sigma(X,Y).
\end{align}

\noindent Thus we see that $(\Sigma,\sigma)\hookrightarrow (M^3,g)$ has CMC,
\begin{align}
H=\frac{2}{3}\mathfrak{H}.\label{sigmaCMC}
\end{align}

\indent We can now use the Codazzi equation for $(\Sigma,\sigma)\hookrightarrow (M^3,g)$. After contraction and taking into account (\ref{sigmaCMC}) we get
\begin{align}
{}^g\!R(Y,\nu)=0.
\end{align}

\noindent This can be used to prove that $\nu(N)$ is constant on $\Sigma$. To this end we calculate the Lie derivative $\mathcal{L}_X(\nu(N))$,
\begin{align}
\mathcal{L}_X(\nu(N))=&2({}^g\!\nabla^2N)(X,\nu)\\
=&2N({}^g\!R(X,\nu)-2\,{}^g\!\nabla_i\varphi\,{}^g\!\nabla_j\varphi X^i\nu^j- N^{-2}\,{}^g\!\nabla_k\Phi\,{}^g\!\nabla^k\Phi e^{-2\alpha\varphi}g_{ij}(X,\nu)\notag\\
&+2N^{-2}e^{-2\alpha\varphi}\,{}^g\!\nabla_i\Phi X^i\,{}^g\!\nabla_j\Phi \nu^j)\notag\\
=&0.\notag
\end{align}

\indent For the function $N:M^3\longrightarrow \mathbb{R}$ and the embedding $(\Sigma,\sigma)\hookrightarrow (M^3,g)$ we have
\begin{align}
{}^g\!\Delta N={}^{\sigma}\!\Delta N+{}^g\!\nabla^2N(\nu,\nu)+ {}^{\sigma}\!\mathrm{tr}(h)\nu(N).
\end{align}

\noindent This can be used in combination with the contracted Gauss equation for $(\Sigma,\sigma)\hookrightarrow (M^3,g)$,
\begin{align}
{}^g\!R-2\,{}^g\!R(\nu,\nu)={}^\sigma\!R-\frac{H^2}{2},
\end{align}

\noindent and the field equations (\ref{emsr1}-\ref{emsr4}) to yield
\begin{align}
N\,{}^{\sigma}\!R=-2N({}^{\mathfrak{g}}\!\nabla_{\nu}\varphi)^2+2H\nu(N)+ \frac{1}{2}NH^2+ 2N^{-1}e^{-2\alpha\varphi}E_{\nu}^2.\label{preint1}
\end{align}

\noindent Now we integrate (\ref{preint1}) on $\Sigma$ and use the Gauss-Bonnet theorem, which leads to
\begin{align}
N_0=\frac{1}{4\pi}N_0^{-1}e^{-2\alpha\varphi_0}E_{\nu}^2A_{\Sigma}+\frac{1}{16\pi}N_0H^2A_{\Sigma}+\frac{1}{4\pi}H[\nu(N)]_0A_{\Sigma}-\frac{1}{4\pi}N_0 ({}^{\mathfrak{g}}\!\nabla_{\nu}\varphi)^2A_{\Sigma}.\label{vajno1}
\end{align}

\indent To derive the next formula we start by applying the contracted Gauss equation to $(\Sigma,\sigma)\hookrightarrow(P^3,p)$ with unit normal $\eta$. This gives
\begin{align}
{}^p\!R+2\,{}^p\!R(\eta,\eta)={}^{\sigma}\!R.
\end{align}

\noindent For the metric (\ref{statg}) ${}^p\!R(\eta,\eta)=0$ and remembering (\ref{CSCP}) we arrive at
\begin{align}
{}^{\sigma}\!R=\frac{3}{2}H^2-2({}^{\mathfrak{g}}\!\nabla_{\nu}\varphi)^2+e^{-2\alpha\varphi}\frac{2}{N^2}E_{\nu}^2.
\end{align}

\noindent Once again we integrate on $\Sigma$,
\begin{align}
8\pi=\frac{3}{2}H^2A_{\Sigma}-2({}^{\mathfrak{g}}\!\nabla_{\nu}\varphi)^2A_{\Sigma}+\frac{2}{N_0^2}e^{-2\alpha\varphi_0}E_{\nu}^2A_{\Sigma}.\label{vajno2}
\end{align}

\noindent From (\ref{vajno1}) and (\ref{vajno2}) we get
\begin{align}\label{vajno3}
2[\nu(N)]_0=\frac{2}{\rho_{0}}=N_0 H
\end{align}

\noindent and
\begin{align}
N_0=\frac{1}{4\pi}e^{-2\alpha\varphi_0}N_0^{-1}E_{\nu}^2A_{\Sigma}-\frac{1}{4\pi}N_0({}^{\mathfrak{g}}\!\nabla_{\nu}\varphi)^2A_{\Sigma} +\frac{3}{8\pi}H [\nu(N)]_0 A_{\Sigma}.
\end{align}

\noindent From (\ref{vajno2}) and (\ref{vajno3}) we can derive one more useful relation

\begin{align}\label{vajno4}
1= \frac{1}{4\pi} H^2 A_{\Sigma} - \frac{1}{4\pi}\left[N^{-2}_{0} [\nu(N)]^2_0 + ({}^{\mathfrak{g}}\!\nabla_{\nu}\varphi)^2  - e^{-2\alpha\varphi_0}N_0^{-2}E_{\nu}^2 \right]  A_{\Sigma}.
\end{align}

\section{Symmetries of the dimensionally reduced EMD equations, divergence identities and functional dependence between the potentials  }

\indent In order to make the symmetries of the dimensionally reduced equations more transparent we  rewrite equations (\ref{emsr1}-\ref{emsr4}) using a new 3-metric $\gamma_{ij}$ on $M^3_{\mathrm{ext}}$,
\begin{align}\label{ConfTM}
\gamma_{ij}=N^2g_{ij},
\end{align}

\noindent and a new function $u$ such that $N^2=e^{2u}$. We get the following equations:
\begin{align}
&{}^{\gamma}\!R_{ij}=2D_iuD_ju+2D_i\varphi D_j\varphi-2e^{-2u-2\alpha\varphi}D_i\Phi D_j\Phi,\\
&D_iD^iu=e^{-2u-2\alpha\varphi}D_i\Phi D^i\Phi,\\
&D_iD^i\varphi=\alpha e^{-2u-2\alpha\varphi}D_i\Phi D^i\Phi,\\
&D_i(e^{-2u-2\alpha\varphi}D^i\Phi)=0,
\end{align}

\noindent where $D$ denotes the covariant derivative in the metric $\gamma$. An even more convenient form of the equations can be obtained if we use the following potentials:
\begin{align}
U=u+\alpha \varphi,\ \Psi=\varphi -\alpha u,\ \hat{\Phi}=\sqrt{1+\alpha^2}\Phi.
\end{align}

\noindent Then the field equations become
\begin{align}
&{}^{\gamma}\!R_{ij}=\frac{1}{1+\alpha^2}(2D_iUD_jU-2e^{-2U}D_i\hat{\Phi}D_j\hat{\Phi}+2D_i\Psi D_j\Psi),\label{emsrr1}\\
&D_iD^iU=e^{-2U}D_i\hat{\Phi}D^i\hat{\Phi},\label{emsrr2}\\
&D_iD^i\Psi=0,\label{emsrr3}\\
&D_i(e^{-2U}D^i\hat{\Phi})=0.\label{emsrr4}
\end{align}

\indent The above equations can be regarded as a 3-dimensional gravity coupled to a non-linear $\sigma$-model parameterized by the scalar fields $\phi^A=(U,\Psi,\hat{\Phi})$ with a target space  metric
\begin{align}\label{TSMETRIC}
G_{AB}d\phi^A d\phi^B=\frac{1}{1+\alpha^2}(\mathrm{d}U^2+\mathrm{d}\Psi^2- e^{-2U}\mathrm{d}\hat{\Phi}^2).
\end{align}

\noindent The Killing vectors for this metric are
\begin{align}
&K^{(1)}=-2\hat{\Phi}\frac{\partial}{\partial U}-e^{2U}(e^{-2U}\hat{\Phi}^2+1)\frac{\partial}{\partial \hat{\Phi}},\label{kill1}\\
&K^{(2)}=-\frac{\partial}{\partial U}-\hat{\Phi}\frac{\partial}{\partial \hat{\Phi}},\label{kill2}\\
&K^{(3)}=-\frac{\partial}{\partial \hat{\Phi}},\label{kill3}\\
&K^{(4)}=\frac{\partial}{\partial \Psi}.\label{kill4}
\end{align}
The corresponding Killing one-forms are given by
\begin{align}
&K^{(1)}_{A}d\phi^A =- 2\hat{\Phi}dU + (1 + e^{-2U}\hat{\Phi}^2)d\hat{\Phi},\label{killf1} \\
&K^{(2)}_{A}d\phi^A= - dU +  e^{-2U}\hat{\Phi}d\hat{\Phi},\label{killf2} \\
&K^{(3)}_{A}d\phi^A= e^{-2U}d\hat{\Phi}, \label{killf3}\\
&K^{(4)}_{A}d\phi^A=d\Psi. \label{killf4}
\end{align}

\indent Using the fact that $K^{(a)}_A d\phi^A $ are Killing one-forms for the metric $G_{AB}$ and taking into account the  equations for $\phi^A$ one can show that the following divergence identities are satisfied:
\begin{align}
D_i(K^{(a)}_AD^i\phi^A)=0.
\end{align}

\noindent  Integrating these equations on $M^3_{\mathrm{ext}}$ is straightforward for each of the Killing one-forms (\ref{killf1} - \ref{killf4}) and taking into account the asymptotic behaviour of the potentials $\phi^A$  we obtain one new functional relation between the potentials on $\Sigma$,
\begin{align}
e^{2U_0}=1+\hat{\Phi}_0^2-\frac{2(M+\alpha q)}{Q_{\alpha}}\hat{\Phi}_0,\label{UPhirel}
\end{align}

\noindent where $Q_{\alpha}=\sqrt{1+\alpha^2}Q$, in addition to the already known relations (\ref{mass},\ref{sccharge}).\\
\indent We will next prove that relation (\ref{UPhirel}) holds on the whole $M^3_{\mathrm{ext}}$. To do this consider the equality
\begin{align}
e^{2U}\omega_i\omega^i=D_i\left[\left( e^{2U}-1+\hat{\Phi}\left( -\hat{\Phi}+\frac{2M}{Q_{\alpha}} +\frac{2\alpha q}{Q_{\alpha}} \right) \right)\omega^i\right],
\end{align}

\noindent which follows from the field equations and where

\begin{eqnarray}
\omega_i=2D_iU+2e^{-2U}\left( -\hat{\Phi}+\frac{M}{Q_{\alpha}}+\frac{\alpha q}{Q_{\alpha}} \right)D_i\hat{\Phi}.
\end{eqnarray}

Integrating the above we get
\begin{align}
\int_{M^3_{\mathrm{ext}}}e^{2U}\omega_i\omega^i\mathrm{d}\mu=&\int_{M^3_{\mathrm{ext}}}D_i\left[\left( e^{2U}-1+\hat{\Phi}\left( -\hat{\Phi}+\frac{2M}{Q_{\alpha}} +\frac{2\alpha q}{Q_{\alpha}} \right) \right)\omega^i\right]\mathrm{d}\mu\\
=&\int_{S^2_{\infty}}\left( e^{2U}-1+\hat{\Phi}\left( -\hat{\Phi}+\frac{2M}{Q_{\alpha}} +\frac{2\alpha q}{Q_{\alpha}} \right) \right)\omega^i\mathrm{d}\Sigma_i\notag\\
&-\int_{\Sigma}\left( e^{2U}-1+\hat{\Phi}\left( -\hat{\Phi}+\frac{2M}{Q_{\alpha}} +\frac{2\alpha q}{Q_{\alpha}} \right) \right)\omega^i\mathrm{d}\Sigma_i\notag\\
=&0,\notag
\end{align}

\noindent where we have used the asymptotic behaviour of the potentials and equation (\ref{UPhirel}). It follows that $\omega_i=0$ on $M^3_{\mathrm{ext}}$ and thus
\begin{align}
e^{2U}-1-\hat{\Phi}^2+\frac{2( M+\alpha q)}{Q_{\alpha}}\hat{\Phi}=0.\label{UPhirelglob}
\end{align}

\indent The next step is to obtain yet another relation between the potentials. To do this we introduce a new potential $\zeta$, such that
\begin{align}\label{DefZeta}
\mathrm{d}\zeta=- e^{-2U}\mathrm{d}\hat{\Phi}, \ \ \zeta_{\infty}=0 .
\end{align}

\noindent Since $U$ and $\hat{\Phi}$ are not independent, $\zeta$ can be used to simplify the field equations even more. We can use $\frac{\omega_i}{2}=0$ and (\ref{UPhirelglob}) to show that
\begin{align}
D_iUD_jU=\left[ e^{-2U}-e^{-4U}+e^{-4U}\left( \frac{M}{Q_{\alpha}}+\frac{\alpha q}{Q_{\alpha}} \right)^2 \right]D_i\hat{\Phi}D_j\hat{\Phi}.
\end{align}

\noindent Then, with the new potential $\zeta$, equations (\ref{emsrr1}-\ref{emsrr4}) take the following form:
\begin{align}
&{}^{\gamma}\!R_{ij}=\frac{2}{1+\alpha^2}\left[ D_i\Psi D_j\Psi +\left[\left( \frac{M}{Q_{\alpha}} +\frac{\alpha q}{Q_{\alpha}} \right)^2-1\right]D_i\zeta D_j\zeta \right],\\
&D_iD^i\Psi=0,\\
&D_iD^i\zeta=0.
\end{align}

\indent Now let $J_i=\zeta D_i \Psi-\Psi D_i \zeta$. From the field equations follows that $D_i J^i=0$. Integrating this over $M^3_{\mathrm{ext}}$  with the asymptotic behavior of
$\zeta$ and $\Psi$ in mind,  we get
\begin{align}
\zeta_0(q_{0} -\alpha M_0)-\Psi_0Q_{\alpha}=0
\end{align}

\noindent on $\Sigma$. This can be extended to the whole of $M^3_{\mathrm{ext}}$ just like (\ref{UPhirel}). In this case we use $\omega_i=(\alpha M_0-q_0)D_i\zeta-Q_{\alpha}D_i\Psi$, which satisfies
\begin{align}
\omega_i\omega^i=D_i\left[ \left[(\alpha M_0-q_0)\zeta-Q_{\alpha}\Psi\right] \omega^i\right].
\end{align}

\noindent After integration we find that $\omega_i=0$ which gives the desired functional dependence, namely
\begin{align}
(q_0 - \alpha M_0)\zeta-Q_{\alpha}\Psi=0.\label{zetaPsirelglob}
\end{align}

\indent Using (\ref{zetaPsirelglob}) we can rewrite the field equations in the form
\begin{align}\label{DRFE1}
&{}^{\gamma}\!R_{ij}=\frac{2}{1+\alpha^2}\left( \frac{M^2+q^2}{Q^2}-1 \right)D_i\zeta D_j\zeta,\\
&D_iD^i\zeta=0.
\end{align}

It is not difficult to see that  $\phi^A(\zeta)=(U(\zeta),\Psi(\zeta), \hat{\Phi}(\zeta))$ is a geodesic of the metric (\ref{TSMETRIC}) with

\begin{eqnarray}\label{geodesic}
G_{AB}\frac{d\phi^A}{d\zeta} \frac{d\phi^B}{d\zeta}= \frac{1}{1+ \alpha^2} \left( \frac{M^2+q^2}{Q^2}-1\right).
\end{eqnarray}

Depending on the ratio $\frac{M^2 + q^2}{Q^2}$ we have three types of geodesics, which will be called  "timelike" for  $M^2 + q^2>Q^2$,
"null" for $M^2 + q^2=Q^2$ and "spacelike" for $M^2 + q^2< Q^2$.

\section{Classification of EMD spacetimes with a photon sphere}

In our previous paper \cite{Yazadjiev2015a} we defined the notion of a non-extremal photon sphere. A photon sphere is non-extremal if $\frac{1}{4\pi} H^2 A_{\Sigma}\ne 1$.
In the case of EMD gravity this condition, as it can easily be shown, is equivalent to $M^2 + q^2 - Q^2 \ne 0$. Here we shall consider only non-extremal photon spheres.
The first main result of the present paper is the following theorem:

\begin{theorem}
Let $(\mathfrak{L}_{ext}^4,\mathfrak{g}, F,\varphi)$ be a static  and asymptotically flat spacetime with  given mass $M$, electric charge $Q$ and dilaton charge $q$, satisfying the Einstein-Maxwell-dilaton equations and possessing a non-extremal photon sphere as an inner boundary of $\mathfrak{L}_{ext}^4$. Assume that the lapse function regularly foliates $\mathfrak{L}_{ext}^4$. Then $(\mathfrak{L}_{ext}^4,\mathfrak{g}, F, \varphi)$ is spherically symmetric.
\end{theorem}

\begin{proof}In proving the theorem we shall follow \cite{Yazadjiev2015} and \cite{Yazadjiev2015a} with some technical modifications.

\vskip 0.3cm

{\bf Case $M^2 + q^2>Q^2 $}

\vskip 0.2cm

\noindent  We  first consider the case corresponding to "timelike" geodesics of the target space metric  when  $M^2+q^2-Q^2>0$ and introduce a new potential $\lambda$, such that
\begin{align}
\lambda=\sqrt{\left( \frac{M^2+q^2}{Q^2}-1 \right)\frac{1}{1+\alpha^2}} \zeta .
\end{align}
From the definition of $\lambda$  and eq. (\ref{DefZeta}) it is easy one to show that $\lambda$ has the following asymptotic behaviour:

\begin{eqnarray}\label{ASLAMBDA}
\lambda = - \frac{\sqrt{M^2 + q^2 -Q^2}}{r} + O(r^{-2}).
\end{eqnarray}

The fact that $\lambda$ is harmonic with the above asymptotic behaviour shows that $\lambda<0$ on  $M^3_{\mathrm{ext}}$.

\noindent Now we will use the inequalities \cite{Yazadjiev2015,Yazadjiev2015a}
\begin{align}\label{INQ1}
\int_{M^3_{\mathrm{ext}}}D^i\left[ \Omega^{-1}(\Gamma D_i\chi - \chi D_i\Gamma) \right]\sqrt{\gamma}\mathrm{d}^3x \ge 0
\end{align}

\noindent and
\begin{align}\label{INQ2}
\int_{M^3_{\mathrm{ext}}}D^i\left( \Omega^{-1}D_i\chi \right)\sqrt{\gamma}\mathrm{d}^3x\ge \int_{M^3_{\mathrm{ext}}}D^i\left[ \Omega^{-1}(\Gamma D_i\chi - \chi D_i\Gamma) \right]\sqrt{\gamma}\mathrm{d}^3x,
\end{align}

\noindent where
\begin{align}
\chi=\left( \gamma^{ij}D_i\Gamma D_j\Gamma \right)^{\frac{1}{4}}, \ \ \Gamma= - \tanh(\lambda), \ \ \Omega=\frac{1}{\cosh^2(\lambda)}.
\end{align}

What is important here is the fact that the equalities in (\ref{INQ1}) and (\ref{INQ2}) hold if and only if the Bach tensor $R(\gamma)_{ijk}$ vanishes.

After long algebra and with the help of the Gauss theorem one can show that the first inequality (\ref{INQ1}) is equivalent to

\begin{eqnarray}
\left(\frac{d\ln(N)}{d\lambda}\right)_{0} \ge -\frac{1}{2} \coth(\lambda_{0}) ,
\end{eqnarray}
while the second inequality (\ref{INQ2}) gives

\begin{eqnarray}\label{INQ2_1}
e^{2\lambda_{0}}\ge \left(M^2 + q^2 - Q^2 \right)^{-1/2} \frac{N^2_{0} \rho_{0} \left(\frac{d\ln(N)}{d\lambda}\right)_{0} }{\left(2\left(\frac{d\ln(N)}{d\lambda}\right)_{0} +1\right)^2}.
\end{eqnarray}

Here the subscript "0" means that the corresponding quantity is evaluated on the photon sphere. In order to further simplify (\ref{INQ2_1}) we
make use of the following equality:

\begin{eqnarray}
N^2_{0}\rho_{0} \left(\frac{d\ln(N)}{d\lambda}\right)_{0}= \left( M^2 + q^2 - Q^2 \right)^{1/2} \left[4 \left(\frac{d\ln(N)}{d\lambda}\right)_{0}^2 - 1\right],
\end{eqnarray}
which can be derived by using the results in section 3 and (\ref{geodesic}). With this equality taken into account, the inequality  (\ref{INQ2}) becomes

\begin{eqnarray}
\left(\frac{d\ln(N)}{d\lambda}\right)_{0} \le -\frac{1}{2} \coth(\lambda_{0}).
\end{eqnarray}

Therefore we have

\begin{eqnarray}\label{PSEq}
\left(\frac{d\ln(N)}{d\lambda}\right)_{0} = -\frac{1}{2} \coth(\lambda_{0})
\end{eqnarray}
and we conclude that the Bach tensor vanishes, $R(\gamma)_{ijk}=0$. This means that the metric $\gamma_{ij}$ is conformally flat. As a direct consequence from
(\ref{ConfTM}) it follows that the metric $g_{ij}$ is also conformally flat or equivalently

\begin{eqnarray}
R(g)_{ijk}=0.
\end{eqnarray}

Since the lapse function $N$ regularly foliates $M^3_{\mathrm{ext}}$, we can write the metric $g_{ij}$ in the form

\begin{eqnarray}
g= \rho^2 dN^2 + \sigma_{AB}dx^A dx^B,
\end{eqnarray}
where $\sigma_{AB}$ is the 2-dimensional metric on the 2-dimensional intersections  $\Sigma_{N}$ of the level
sets $N = const$ with $M^3_{\mathrm{ext}}$. Let us denote the second fundamental form of $(\Sigma_N,\sigma)\hookrightarrow(M^3,g)$ by
$h^{\Sigma_{N}}_{AB}$ and its trace by $H^{\Sigma_{N}}$. After long calculations and with the help of the dimensionally reduced field equations
we find

\begin{eqnarray}
R(g)_{ijk}R(g)^{ijk}= \frac{8(1+ \alpha^2)^2}{N^4\rho^4}\frac{ \left(\frac{M^2+ q^2}{Q^2}-1\right)^2} {\left[\sqrt{e^{2U} + \left(\frac{M + \alpha q}{Q_{\alpha}}\right)^2 -1} +\alpha
\left(\frac{\alpha M - q}{Q_{\alpha}}\right) \right]^4}  \times \\\nonumber  \left[\left(h^{\Sigma_{N}}_{AB}- \frac{1}{2}H^{\Sigma_{N}} \sigma_{AB}\right)
\left(h^{\Sigma_{N}\, AB}- \frac{1}{2}H^{\Sigma_{N}} \sigma^{AB} \right)    + \frac{1}{2\rho^2}\sigma^{AB}\partial_{A}\rho\partial_{B}\rho \right].
\end{eqnarray}
Since $R(g)_{ijk}=0$ and $M^2 +q^2> Q^2$ we can conclude that

\begin{eqnarray}
h^{\Sigma_{N}}_{AB}= \frac{1}{2}H^{\Sigma_{N}} \sigma_{AB},\;\;\; \partial_{A}\rho=0.
\end{eqnarray}

Therefore the metric $g_{ij}$ is spherically symmetric. The same applies to the metric $\gamma_{ij}$.
Let us also note that Eq. (\ref{PSEq}) is just the equation for the photon sphere for $M^2 + q^2>Q^2 $ and it arises naturally in our approach.

\vskip 0.3cm

{\bf Case $M^2 + q^2<Q^2 $}

\vskip 0.2cm

\noindent In the case under consideration we will use the potential $ \lambda=\sqrt{\left(1-  \frac{ M^2+q^2}{Q^2} \right)\frac{1}{1+\alpha^2}} \zeta$ with asymptotic

\begin{eqnarray}\label{ASLAMBDA1}
\lambda = - \frac{\sqrt{Q^2 - M^2 - q^2}}{r} + O(r^{-2}).
\end{eqnarray}

One can also show that $-\frac{\pi}{2}<\lambda<0$ on  $M^3_{\mathrm{ext}}$. As in the previous case we consider the inequalities (\ref{INQ1}) and (\ref{INQ2}), but this time with
different functions $\Gamma$ and $\Omega$, namely

\begin{eqnarray}
\Gamma= - \tan(\lambda), \; \; \; \Omega= \cos^{-2}(\lambda).
\end{eqnarray}
Following the same steps as in the previous case one can show that the first inequality  reduces to

\begin{eqnarray}
\left(\frac{d\ln(N)}{d\lambda}\right)_{0}\ge -\frac{1}{2} \cot(\lambda_{0})
\end{eqnarray}
while the second inequality (\ref{INQ2}) gives

\begin{eqnarray}
\left(\frac{d\ln(N)}{d\lambda}\right)_{0}\le -\frac{1}{2} \cot(\lambda_{0}).
\end{eqnarray}
Hence we conclude that
\begin{eqnarray}\label{PSEq1}
\left(\frac{d\ln(N)}{d\lambda}\right)_{0}= -\frac{1}{2} \cot(\lambda_{0})
\end{eqnarray}
and therefore $R(\gamma)_{ijk}=0$ which means that
 $\gamma_{ij}$ and $g_{ij}$ are conformally flat. The same argument as in the previous case shows that $g_{ij}$ is spherically symmetric.
Let us also note that (\ref{PSEq1}) is the equation for the photon sphere in the case under consideration.
\end{proof}

The second main result of this paper is the explicit classification of the static and asymptotically flat Einstein-Maxwell-dilaton spacetimes possessing a photon sphere.
In order to simplify the equations  we will use a new parameter $M_{\alpha}$, defined by

\begin{eqnarray}
M_{\alpha}=M+ \alpha q.
\end{eqnarray}
It is also useful to give the following formula
\begin{eqnarray}
M^2 + q^2 - Q^2 = \frac{1}{1+ \alpha^2}\left[M^2_{\alpha}- Q^2_{\alpha} + \left(q-\alpha M\right)^2 \right].
\end{eqnarray}
The derivation of the solutions, possessing a photon sphere, is as follows.

\medskip
\noindent

{\bf Case $M^2 + q^2>Q^2 $}

\medskip
\noindent

In this case the dimensionally reduced field equations (\ref{DRFE1}) become

\begin{eqnarray}
&&{}^{\gamma}\!R_{ij}=2D_i\lambda D_j\lambda, \nonumber \\
&&D_iD^i\lambda=0.
\end{eqnarray}
These equations are in fact the static vacuum Einstein equations written in terms of the metric $\gamma_{ij}$ with an effective lapse
function $N_{eff}=e^{\lambda}$ having an effective mass $M_{eff}=\sqrt{M^2 + q^2-Q^2}$ as follows from (\ref{ASLAMBDA}). Since
the Schwarzschild solution is the only static and spherically symmetric solution to the vacuum Einstein equations, we find

\begin{eqnarray}
&&e^{2\lambda}= 1- \frac{2\sqrt{M^2 + q^2 -Q^2}}{r},\\
&&\gamma_{ij}dx^idx^j= dr^2 +e^{2\lambda} r^2(d\theta^2 + \sin^2\theta d\phi^2).
\end{eqnarray}

The spacetime metric is therefore

\begin{eqnarray}
ds^2 = -N^2 dt^2 + N^{-2}\left[dr^2 +e^{2\lambda} r^2(d\theta^2 + \sin^2\theta d\phi^2)\right].
\end{eqnarray}

In order to obtain the  lapse function $N$, the electrostatic potential $\Phi$ and the dilaton field $\varphi$ we have to integrate eq.(\ref{DefZeta})  and take into account
eq.(\ref{zetaPsirelglob}). Depending on $M_{\alpha}$ and $Q_{\alpha}$ we have three classes of solutions.

\medskip
\noindent

1) The first class of  solutions is obtained for $M^2_{\alpha}>Q^2_{\alpha}$ and  the lapse function $N$, the electrostatic potential $\Phi$ and the dilaton field $\varphi$ are given by

\begin{eqnarray}
&&N^2= \left[4\left(1-\frac{Q^2_{\alpha}}{M^2_{\alpha}} \right)\right]^{\frac{1}{1+ \alpha^2}} \frac{ e^{\frac{2}{1+ \alpha^2} \frac{\sqrt{ M^2_{\alpha}- Q^2_{\alpha}}- \alpha\left(q-\alpha M\right) }{\sqrt{M^2 + q^2 -Q^2}} \lambda } }{\left[1 +  \sqrt{1 -  \frac{Q^2_{\alpha}}{M^2_{\alpha}}} - \left( 1 - \sqrt{1-  \frac{Q^2_{\alpha}}{M^2_{\alpha}} }\right) e^{2\sqrt{\frac{M^2_{\alpha}- Q^2_{\alpha}}{M^2 + q^2 -Q^2} }\lambda }\right]^{\frac{2}{1+ \alpha^2}} },
\nonumber \\
&&\Phi= \frac{Q}{M_{\alpha}} \frac{1 - e^{2\sqrt{\frac{M^2_{\alpha}- Q^2_{\alpha}}{M^2 + q^2 -Q^2} }\lambda } }{1 +  \sqrt{1 -  \frac{Q^2_{\alpha}}{M^2_{\alpha}}} - \left( 1 - \sqrt{1-  \frac{Q^2_{\alpha}}{M^2_{\alpha}} }\right) e^{2\sqrt{\frac{M^2_{\alpha}- Q^2_{\alpha}}{M^2 + q^2 -Q^2} }\lambda }},\\
&&e^{2\varphi}= \left[4\left(1 - \frac{Q^2_{\alpha}}{M^2_{\alpha}}\right)\right]^{\frac{\alpha}{1+ \alpha^2}} \frac{ e^{\frac{2}{1+ \alpha^2} \frac{\alpha \sqrt{ M^2_{\alpha}- Q^2_{\alpha}}+ \left(q-\alpha M\right) }{\sqrt{M^2 + q^2 -Q^2}} \lambda } }{\left[1 +  \sqrt{1 -  \frac{Q^2_{\alpha}}{M^2_{\alpha}}} - \left( 1 - \sqrt{1-  \frac{Q^2_{\alpha}}{M^2_{\alpha}} }\right) e^{2\sqrt{\frac{M^2_{\alpha}- Q^2_{\alpha}}{M^2 + q^2 -Q^2} }\lambda }\right]^{\frac{2\alpha}{1+ \alpha^2}} }.\nonumber
\end{eqnarray}
It is not difficult to show that for this class of solutions the equation for the photon sphere (\ref{PSEq}) has solutions only when the parameters  $M$, $Q$ and $q$ are subject to the inequality

\begin{eqnarray}
\sqrt{M^2_{\alpha} - Q^2_{\alpha}} + \alpha\left(\alpha M  - q\right)> \frac{1}{2}\sqrt{1+ \alpha^2} \sqrt{M^2_{\alpha} - Q^2_{\alpha}  + (q - \alpha M)^2  }.
\end{eqnarray}

An important subclass of solutions are the black hole solutions. The EMD black hole solutions correspond to $\alpha\sqrt{M^2_{\alpha} - Q^2_{\alpha}}=\alpha M - q$. It is easy to see
that the above inequality is satisfied for the black hole solutions and therefore the EMD black holes always possess a photon sphere.

2) The second class of solutions is obtained for $M^2_{\alpha}=Q^2_{\alpha}$ and we have

\begin{eqnarray}
&&N^2= \frac{e^{- \frac{2\varepsilon\alpha}{\sqrt{1+ \alpha^2}}\lambda}}{\left[1+ \varepsilon\sqrt{1+ \alpha^2}\frac{M+ \alpha q}{\alpha M - q}\lambda\right]^{\frac{2}{1+ \alpha^2}} },\nonumber \\
&&\Phi= \frac{\varepsilon \frac{M+ \alpha q}{\alpha M-q} \lambda }{1+ \varepsilon\sqrt{1+ \alpha^2}\frac{M+ \alpha q}{\alpha M - q}\lambda},\\
&&e^{2\varphi}= \frac{e^{- \frac{2\varepsilon}{\sqrt{1+ \alpha^2}}\lambda}}{\left[1+ \varepsilon\sqrt{1+ \alpha^2}\frac{M+ \alpha q}{\alpha M -q}\lambda\right]^{\frac{2\alpha}{1+ \alpha^2}} }, \nonumber
\end{eqnarray}
where $\varepsilon=(q-\alpha M)/|q-\alpha M|$. The equation for the sphere (\ref{PSEq})  has solutions only when  $M$ and $q$ are subject to the inequality

\begin{eqnarray}
\alpha(\alpha M -q)> \frac{1}{2}\sqrt{1+ \alpha^2} |q-\alpha M|.
\end{eqnarray}

3) The third class of solutions is for $M^2_{\alpha}<Q^2_{\alpha}$ and we have

\begin{eqnarray}
&&N^2= \left(1 - \frac{M^2_{\alpha}}{Q^2_{\alpha}}\right)^{\frac{1}{1+\alpha^2}} \frac{e^{\frac{2\alpha}{1+ \alpha^2} \frac{(\alpha M - q)}{\sqrt{M^2 + q^2 - Q^2}}\lambda }}
{\cos^{\frac{2}{1 + \alpha^2}}\left[\sqrt{\frac{Q^2_{\alpha}-M^2_{\alpha}}{M^2 +q^2 -Q^2}} \lambda + \arctan\left(\frac{\frac{M_{\alpha}}{Q_{\alpha}}}{\sqrt{1 - \frac{M^2_{\alpha}}{Q^2_{\alpha}}} }\right)\right]},  \\
&&\Phi= \frac{1}{\sqrt{1+ \alpha^2}}\left[\frac{M_{\alpha}}{Q_{\alpha}} - \sqrt{1 - \frac{M^2_{\alpha}}{Q^2_{\alpha}}} \tan\left(\sqrt{\frac{Q^2_{\alpha}-M^2_{\alpha}}{M^2 +q^2 -Q^2}} \lambda + \arctan\left(\frac{\frac{M_{\alpha}}{Q_{\alpha}}}{\sqrt{1 - \frac{M^2_{\alpha}}{Q^2_{\alpha}}} }\right)\right)\right], \nonumber\\
&&e^{2\varphi}= \left(1 - \frac{M^2_{\alpha}}{Q^2_{\alpha}}\right)^{\frac{\alpha}{1+\alpha^2}} \frac{e^{-\frac{2}{1+ \alpha^2} \frac{(\alpha M - q)}{\sqrt{M^2 + q^2 - Q^2}}\lambda }}
{\cos^{\frac{2\alpha}{1 + \alpha^2}}\left[\sqrt{\frac{Q^2_{\alpha}-M^2_{\alpha}}{M^2 +q^2 -Q^2}} \lambda + \arctan\left(\frac{\frac{M_{\alpha}}{Q_{\alpha}}}{\sqrt{1 - \frac{M^2_{\alpha}}{Q^2_{\alpha}}} }\right)\right]}.\nonumber
\end{eqnarray}
The equation for the photon sphere (\ref{PSEq}) possesses solutions only for  $M$, $Q$ and $q$ subject to the inequality

\begin{eqnarray}
&&\alpha (\alpha M - q) + \sqrt{Q^2_{\alpha}-M^2_{\alpha}} \tan\left[\sqrt{1+ \alpha^2}\frac{\sqrt{Q^2_{\alpha}-M^2_{\alpha}}}{\sqrt{M^2_{\alpha} - Q^2_{\alpha} + (q-\alpha M)^2 }} + \arctan\left(\frac{M_{\alpha}}{\sqrt{Q^2_{\alpha}-M^2_{\alpha}}}\right)\right] \nonumber \\
&&\ge -\frac{1}{2}\sqrt{1+\alpha^2}  \sqrt{M^2_{\alpha} - Q^2_{\alpha} + (q-\alpha M)^2} \coth(\lambda_c),
\end{eqnarray}
where the critical value  $\lambda_c$ of $\lambda$ is given by

\begin{eqnarray}
\tanh(\lambda_c)= \frac{\sqrt{M^2_{\alpha} -Q^2_{\alpha} + (q-\alpha M)^2}}{Q^2_{\alpha} - M^2_{\alpha} + (1+ 2\alpha^2)(q-\alpha M)^2}
\left\{\alpha\sqrt{1+ \alpha^2}(\alpha M-q)  \right. \nonumber\\
\left. -\frac{1}{\sqrt{2}}\sqrt{\alpha^2[M^2_{\alpha} - Q^2_{\alpha} + 3(q-\alpha M)^2] + Q^2_{\alpha} -M^2_{\alpha} + (q-\alpha M)^2}\right\}.
\end{eqnarray}

\medskip
\noindent

{\bf Case $M^2 + q^2< Q^2 $}

\medskip
\noindent

In this case the dimensionally reduced equations become

\begin{eqnarray}
&&{}^{\gamma}\!R_{ij}=-2D_i\lambda D_j\lambda, \nonumber \\
&&D_iD^i\lambda=0.
\end{eqnarray}
Solving these equations for spherically symmetric space and taking into account (\ref{ASLAMBDA1})
we find

\begin{eqnarray}
&&\lambda= \arctan\left(\frac{\sqrt{Q^2-M^2-q^2}}{r}\right),\\
&&\gamma_{ij}= dr^2 + (r^2 + Q^2 - M^2 - q^2)(d\theta^2 + \sin^2\theta d\phi^2).
\end{eqnarray}

The 4-dimensional metric is then
\begin{eqnarray}
ds^2 = - N^2 dt^2 + N^{-2} \left[dr^2 + (r^2 + Q^2 - M^2 - q^2)(d\theta^2 + \sin^2\theta d\phi^2) \right].
\end{eqnarray}

The lapse function $N$, the electrostatic potential $\Phi$ and the dilaton field $\varphi$ are given by

\begin{eqnarray}
&&N^2= \left(1 - \frac{M^2_{\alpha}}{Q^2_{\alpha}}\right)^{\frac{1}{1+\alpha^2}} \frac{e^{\frac{2\alpha}{1+ \alpha^2} \frac{(\alpha M - q)}{\sqrt{Q^2 -M^2 - q^2}}\lambda }}
{\cos^{\frac{2}{1 + \alpha^2}}\left[\sqrt{\frac{Q^2_{\alpha}-M^2_{\alpha}}{Q^2 -M^2 -q^2}} \lambda + \arctan\left(\frac{\frac{M_{\alpha}}{Q_{\alpha}}}{\sqrt{1 - \frac{M^2_{\alpha}}{Q^2_{\alpha}}} }\right)\right]},  \\
&&\Phi= \frac{1}{\sqrt{1+ \alpha^2}}\left[\frac{M_{\alpha}}{Q_{\alpha}} - \sqrt{1 - \frac{M^2_{\alpha}}{Q^2_{\alpha}}} \tan\left(\sqrt{\frac{Q^2_{\alpha}-M^2_{\alpha}}{Q^2 -M^2 -q^2}} \lambda + \arctan\left(\frac{\frac{M_{\alpha}}{Q_{\alpha}}}{\sqrt{1 - \frac{M^2_{\alpha}}{Q^2_{\alpha}}} }\right)\right)\right], \nonumber\\
&&e^{2\varphi}= \left(1 - \frac{M^2_{\alpha}}{Q^2_{\alpha}}\right)^{\frac{\alpha}{1+\alpha^2}} \frac{e^{-\frac{2}{1+ \alpha^2} \frac{(\alpha M - q)}{\sqrt{Q^2- M^2 - q^2}}\lambda }}
{\cos^{\frac{2\alpha}{1 + \alpha^2}}\left[\sqrt{\frac{Q^2_{\alpha}-M^2_{\alpha}}{Q^2 -M^2 -q^2}} \lambda + \arctan\left(\frac{\frac{M_{\alpha}}{Q_{\alpha}}}{\sqrt{1 - \frac{M^2_{\alpha}}{Q^2_{\alpha}}} }\right)\right]}.\nonumber
\end{eqnarray}

The equations for the photon sphere (\ref{PSEq}) possesses solutions only for  $M$, $Q$ and $q$ subject to the inequality

\begin{eqnarray}
&&\alpha (\alpha M - q) + \sqrt{Q^2_{\alpha}-M^2_{\alpha}} \tan\left[\sqrt{1+ \alpha^2}\frac{\sqrt{Q^2_{\alpha}-M^2_{\alpha}}}{\sqrt{Q^2_{\alpha} - M^2_{\alpha} - (q-\alpha M)^2 }} + \arctan\left(\frac{M_{\alpha}}{\sqrt{Q^2_{\alpha}-M^2_{\alpha}}}\right)\right] \nonumber \\
&&\ge -\frac{1}{2}\sqrt{1+\alpha^2}  \sqrt{Q^2_{\alpha} - M^2_{\alpha} - (q-\alpha M)^2} \cot(\lambda_c)
\end{eqnarray}

where the critical value  $\lambda_c$ of $\lambda$ is given by

\begin{eqnarray}
\tan(\lambda_c)= \frac{\sqrt{Q^2_{\alpha} -M^2_{\alpha} - (q-\alpha M)^2}}{Q^2_{\alpha} - M^2_{\alpha} + (1+ 2\alpha^2)(q-\alpha M)^2}
\left\{-\alpha\sqrt{1+ \alpha^2}(\alpha M-q)  \right. \nonumber\\
\left. -\frac{1}{\sqrt{2}}\sqrt{\alpha^2[M^2_{\alpha} - Q^2_{\alpha} + 3(q-\alpha M)^2] + Q^2_{\alpha} -M^2_{\alpha} + (q-\alpha M)^2}\right\}
\end{eqnarray}

with the additional restriction

\begin{eqnarray}
(q - \alpha M)^2 > \frac{\alpha^2-1}{1+ 3\alpha^2} (Q^2_{\alpha} - M^2_{\alpha})
\end{eqnarray}
for $\alpha^2>1$.

This way we proved the following theorem:

\begin{theorem}
Let $(\mathfrak{L}_{ext}^4,\mathfrak{g}, F,\varphi)$ be a static  and asymptotically flat spacetime with  given mass $M$, electric charge charge $Q$ and dilaton charge $q$, satisfying the Einstein-Maxwell-dilaton equations and possessing a non-extremal photon sphere as an inner boundary of $\mathfrak{L}_{ext}^4$. Assume that the lapse function regularly foliates $\mathfrak{L}_{ext}^4$. Then $(\mathfrak{L}_{ext}^4,\mathfrak{g}, F, \varphi)$  is isometric to one of the spherically symmetric solutions described above with $M$, $Q$ and $q$ subject to
the corresponding inequalities given above.
\end{theorem}

\section{Conclusion}

In the present paper we considered the problem for the classification of the static and asymptotically flat EMD spacetimes possessing a photon sphere.
We first proved that photon spheres in EMD gravity have constant mean curvature and constant scalar curvature and  derived some relations between
the mean curvature of the photon sphere and its physical characteristics. Using further the symmetries of the dimensionally reduced EMD equations
we showed that the lapse function, the electrostatic potential and the dilaton field are functionally dependent.  Then, assuming that the photon sphere is non-extremal we proved that the static and asymptotically flat EMD spacetimes  are spherically symmetric. Finally we derived all static and asymptotically  flat  spherically symmetric EMD solutions
with a photon sphere which are fully specified in terms of the mass $M$, the electric charge  $Q$ and the scalar charge $q$ subject to certain inequalities. Our results were derived under the natural technical assumption that  the lapse function regularly foliates the spacetime outside the photon sphere. This assumption and the connectedness of the photon sphere can be easily dropped only in the "black hole case" but not in the general case. In simple words, in the black hole case we can continue the solutions inside the photon sphere(s) to the black hole horizon and  then we can apply the static black hole uniqueness theorem approach \cite{Bunting1987} as shown in \cite{Cederbaum2015a,Cederbaum2015b}. However, this technique, at least in the form used in \cite{Cederbaum2015a,Cederbaum2015b}, can not be applied in the general case (i.e. for non-black hole spacetimes). In most cases we formally have  naked singularities and there is no horizon inside the photon sphere.

The natural generalization of the present paper is to consider higher dimensional spacetimes. With some additional assumptions compared to the 4-dimensional case, our results
can be extended to the case of  higher dimensional EMD gravity, which includes  higher dimensional vacuum Einstein and Einstein-Maxwell equations as particular cases. These results will be
presented elsewhere.

\vskip 0.3cm

\noindent {\bf Acknowledgements:} S. Y. would like to thank  the Research Group Linkage Programme of the Alexander von Humboldt Foundation
for the support. The partial support by the COST Action MP1304, by Bulgarian NSF grant DFNI T02/6 and by Sofia University Research Grant N70/2015 is also gratefully acknowledged.

\end{document}